\documentclass{article}

\usepackage{latexsym}
\usepackage{amsmath}
\usepackage{amsfonts}
\usepackage{amssymb}
\usepackage{amsthm}
\usepackage{graphicx} 
\usepackage{subfigure} 
\usepackage{color}
\usepackage{caption}
\usepackage{indentfirst}  
\usepackage[margin=0.9in]{geometry}
\usepackage{slashbox}  
\usepackage[noend]{algpseudocode}
\usepackage{algorithmicx,algorithm}

\newtheorem{theorem}{Theorem}[section]

\newtheorem{definition}[theorem]{Definition} 

\newtheorem{lemma}[theorem]{Lemma} 

\newtheorem{remark}[theorem]{Remark} 

\newtheorem{proposition}[theorem]{Proposition}  

\newtheorem{corollary}[theorem]{Corollary}  

\begin{document}

\title{\bfseries An analysis of noise folding for low-rank matrix recovery}

\author{
Jianwen Huang$^{1}$,
Jianjun Wang$^{1,~2}$\thanks{Corresponding author, E-mail: wjjmath@gmail.com, wjj@swu.edu.cn(J.J. Wang)},
Feng Zhang$^{1}$,
Hailin Wang$^{1}$,
Wendong Wang$^{1}$
\\
\emph{\small $^{1}$School of Mathematic $\&$ Statistics, Southwest University, Chongqing~400715}\\
\emph{\small $^{2}$School of Artificial Intelligence, Southwest University, Chongqing~400715}
}
\date{}
\maketitle

\subparagraph{Abstract.}
Previous work regarding low-rank matrix recovery has concentrated on the scenarios in which the matrix is noise-free and the measurements are corrupted by noise. However, in practical application, the matrix itself is usually perturbed by random noise preceding to measurement. This paper concisely investigates this scenario and evidences that, for most measurement schemes utilized in compressed sensing, the two models are equivalent with the central distinctness that the noise associated with (\ref{eq.3}) is larger by a factor to $mn/M$, where $m,~n$ are the dimension of the matrix and $M$ is the number of measurements. Additionally, this paper discusses the reconstruction of low-rank matrices in the setting, presents sufficient conditions based on the associating null space property to guarantee the robust recovery and obtains the number of measurements. Furthermore, for the non-Gaussian noise scenario, we further explore it and give the corresponding result. The simulation experiments conducted, on the one hand show effect of noise variance on recovery performance, on the other hand demonstrate the verifiability of the proposed model.
\subparagraph{Key words.}
Compressed sensing; low-rank matrix recovery; noise folding; null space property; restricted isometry property.

\section{Introduction}
\label{sec1}

In recent years, low-rank matrix recovery (LRMR) from noisy measurements, with applications in collaborative filtering \cite{Abernethy et al}, machine learning \cite{Chang et al} \cite{Lin and Li}, control \cite{Xu et al 2018}, quantum tomography \cite{Recht et al}, recommender systems \cite{Mazumder et al}, and remote sensing \cite{Y Wang et al 2017}, has gained significant interest. Formally, this problem considers linear measurements of a (approximately) low-rank matrix $X\in\mathbb{R}^{m\times n}$ of the following form
\begin{align}\label{eq.1}
y=\mathcal{A}(X)+w,
\end{align}
where $y\in\mathbb{R}^M$ is the observed vector, $w\in\mathbb{R}^M$ is an additive noise term, and $\mathcal{A}:\mathbb{R}^{m\times n}\to\mathbb{R}^M$ is a linear measurement map, which is determined by
\begin{align}\label{eq.2}
\mathcal{A}(X)=\left[\mbox{tr}(X^{\top}A^{(1)}),\mbox{tr}(X^{\top}A^{(2)}),\cdots,\mbox{tr}(X^{\top}A^{(M)})\right]^{\top}.
\end{align}
Here, $\mbox{tr}(\cdot)$ is the trace function, $X^{\top}$ is the transposition of $X$ and $A^{(1)}, A^{(2)}, \cdots, A^{(M)}$ are called measurement matrices. Each $A^{(i)}$ can be equal to a row of a compressive measurement matrix, and $\mathcal{A}(X)$ could be written as
\begin{align}\label{eq.6}
\mathcal{A}(X)=\left[\begin{matrix}
\mbox{vec}^{\top}(A^{(1)})\\
\mbox{vec}^{\top}(A^{(2)})\\
\vdots\\
\mbox{vec}^{\top}(A^{(M)})
\end{matrix}\right]\mbox{vec}(X):=A\mbox{vec}(X),
\end{align}
where $\mbox{vec}(X)$ is a long vector gained by stacking the columns of $X$ and $A$ is an $M\times mn$ matrix defined by (\ref{eq.6}) which associates with the linear measurement map $\mathcal{A}$.

However, the aforementioned model (\ref{eq.1}) only considers the noise introduced at the measurement stage. In a variety of application scenarios, the matrix $X$ to be recovered may also be corrupted by noise. Such issue exists in a great number of applications such as the recovery of a video sequence \cite{Candes et al 2009} \cite{Waters et al}, statistical modeling of hyperspectral imaging \cite{Chakrabarti and Zickler}, robust matrix completion \cite{Chen et al 2011}, and signal processing \cite{Arias-Castro and Eldar} \cite{Peter et al}. Accordingly, it is appropriate to take into the following model account
\begin{align}\label{eq.3}
y=\mathcal{A}(X+Z)+w,
\end{align}
where $Z\in\mathbb{R}^{m\times n}$ denotes the noise on the original matrix. Throughout this paper, we suppose that $w$ is a white noise vector satisfying $\mathbb{E}(w)=0_M$ and $\mbox{Var}(w)=\sigma^2I_M$, and similarly $Z$ is a white noise matrix obeying $\mathbb{E}(Z)=0_{m\times n}$ and $\mbox{Var}(Z)=\sigma^2_0I_{mn}$, independent of $w$. Here and elsewhere in this paper, $I_m$ stands for the identity matrix of order $m$. Under these hypotheses, in the next section, we will reveal that the model (\ref{eq.3}) is equivalent to
\begin{align}
\label{eq.31}\tilde{y}=\mathcal{B}(X)+u,
\end{align}
where $\mathcal{B}$ is a linear measurement map, whose restricted isometry property and spherical section property constants are very close to those of $\mathcal{A}$, and $u$ is white noise with mean zero and covariance matrix $(\sigma^2+mn\sigma^2_0/M)I_M$.

When $m=n$ and the matrices $X=\mbox{diag}(x)~(x\in\mathbb{R}^m)$ and $Z=\mbox{diag}(z)~(z\in\mathbb{R}^m)$ are diagonal, the models (\ref{eq.1}) and (\ref{eq.3}) degenerates to the vector models
\begin{align}\label{eq.23}
y=\tilde{A}x+w,
\end{align}
\begin{align}\label{eq.24}
y=\tilde{A}(x+z)+w,
\end{align}
where $\tilde{A}\in\mathbb{R}^{M\times m}$ is the measurement matrix, and $z$ is the noise on the original signal, for more details, see \cite{Candes et al 2006} \cite{Donoho} \cite{Arias-Castro and Eldar} and \cite{Peter et al}. As far as we know, recently most researchers either only discuss the situation of the noise matrix $Z=0$ in (\ref{eq.3}), or merely think over the vector model (\ref{eq.24}) and its associating sparse recovery problem. Specifically, Arias-Castro and Eldar \cite{Arias-Castro and Eldar} considered the model (\ref{eq.24}) and showed that, for the vast majority of measurement schemes employed in compressed sensing, the two models (\ref{eq.23}) and (\ref{eq.24}) are equivalent with the significant distinction
that the signal-to-noise ratio (SNR) is divided by a factor
proportional to $m/M$. For the model (\ref{eq.3}) with $Z=0$, Recht et al. \cite{Recht et al} showed that the minimum-rank solution can be recovered by solving a convex optimization
problem if a certain restricted isometry property holds for the linear transformation defining the
constraints. More related works can be found in \cite{Kong and Xiu} \cite{Zhang and Li 2018b} and \cite{Chen and Li}.

In this paper, our main work incorporates the following parts: firstly, we investigate the relation between the restricted isometry property constants and the spherical section property constants of $\mathcal{B}$ and $\mathcal{A}$ when $\|I_M-(M/mn)AA^{\top}\|<1/2$; secondly, based on certain properties of the null space of the linear measurement map, we establish a sufficient condition for stable and robust recovery of the low-rank matrix itself contaminated by noise and the corresponding upper bound estimation of recovery error; thirdly, we obtain the minimal amount of measurements regarding the sufficient condition guaranteeing recovery via the nuclear norm minimization; finally, the results of numerical experiments show that the method of nuclear norm minimization is effective in recovering low rank matrices after whitening treatment.

The rest of the paper is constructed as follows. In Section \ref{sec2}, we discuss the relationship between the restricted isometry property constants and spherical section property constants of $\mathcal{B}$ and $\mathcal{A}$ under certain conditions. In Section \ref{sec3}, the recovery of low-rank matrices is thought over via the nuclear norm minimization method and sufficient conditions are established to ensure the robust reconstruction. In Section \ref{sec4}, the sampling number based on null space property that make sure the stable recovery is present. Some simulation experiments are carried out in Section \ref{sec5}.
The proofs of the main results are provided in Section \ref{sec6}. Finally, the conclusion is given in Section \ref{sec7}.

\section{RIP and SSP Analysis}
\label{sec2}

In order to derive our results, the model (\ref{eq.3}) can be transformed into
\begin{align}\label{eq.4}
y=\mathcal{A}(X)+v,
\end{align}
where $v$ is determined by
\begin{align}\label{eq.5}
v=\mathcal{A}(Z)+w=A\mbox{vec}(Z)+w.
\end{align}

Due to the assumption of white noise and independence, one can easily verify that the covariance $\Sigma$ of the noise vector $v$ is equal to $\sigma^2I_M+\sigma^2_0AA^{\top}$. Obviously, $v$ is not white noise like the noise $w$, so the recovery analysis may become more complex.

Set $\theta:=\sigma^2+mn\sigma^2_0/M$, $\Sigma_1:=\Sigma/\theta$, $\tilde{y}:=\Sigma^{-1/2}_1y$, $B:=\Sigma^{-1/2}_1A$, $u:=\Sigma^{-1/2}_1v$.
In order to whiten the noise vector $v$, through multiplying the equation (\ref{eq.4}) by $\Sigma^{-1/2}_1$, then we derive the equivalent equation below
\begin{align}\label{eq.7}
\tilde{y}=B\mbox{vec}(X)+u.
\end{align}
 By applying (\ref{eq.6}), the model (\ref{eq.7}) can be written as \begin{align}\label{eq.25}
\tilde{y}=\mathcal{B}(X)+u,
\end{align}
where
\begin{align}\label{eq.8}
\mathcal{B}(X)=\left[\mbox{tr}(X^{\top}B^{(1)}),\mbox{tr}(X^{\top}B^{(2)}),\cdots,\mbox{tr}(X^{\top}B^{(M)})\right]^{\top},
\end{align}
$\mbox{vec}^{\top}(B^{(i)})=B_{i\cdot}=(\Sigma^{-1/2}_1)_{i\cdot}A$, and $B_{i\cdot}$ denotes the $i$th row of the matrix $B$. Observe that the noise vector $u$ is the white noise and its covariance matrix equals to $\theta I_M$. In order to investigate (\ref{eq.3}), we can utilize the results which are exploited to deal with (\ref{eq.1}), with the central distinctness that the noise corresponding with (\ref{eq.3}) is larger by a factor proportional to $mn/M$. In the case of $M\ll mn$, this gives rise to a large noise amplification or \emph{noise folding}. The specific reason is the linear measurement $\mathcal{A}$ amalgamates all the noise entries in $Z$, even those associated to zero entries in $X$, accordingly it brings about a large noise raise in the compressed sampling.

Our analysis depends on approximating $AA^{\top}$ by $(mn/M)I_M$. Set
\begin{align}\label{eq.9}
\delta:=\left\|I_M-\frac{M}{mn}AA^{\top}\right\|.
\end{align}
Here, $\delta$ weighs the quality approximating $AA^{\top}$ by $(mn/M)I_M$ and $\|\cdot\|$ represents the operator norm on $\mathbb{R}^{M\times M}$. For the rest of this paper, suppose that $\delta$ is small. The assumption not only holds with high probability, but also has been shown in \cite{Vershynin}.

In the following, we investigate what is the relationship between the restricted isometry constants of $\mathcal{B}$ and $\mathcal{A}$.

For each integer $r=1,2,\cdots,n_0$, where $n_0=\min\{m,n\}$, we say that a linear measurement map $\mathcal{A}:\mathbb{R}^{m\times n}\to\mathbb{R}^M$ has the restricted isometry property (RIP) with constants $0<\mu_r\leq\nu_r$ if
\begin{align}\label{eq.10}
\mu_r\|X\|^2_F\leq\|\mathcal{A}(X)\|_2^2\leq \nu_r\|X\|^2_F
\end{align}
holds for all matrices $X\in\mathbb{R}^{m\times n}$ of rank at most $r$ (abbreviated as $r$-rank), where $\|X\|_F:=\sqrt{\left<X,X\right>}=\sqrt{\mbox{tr}(X^{\top}X)}$. The theorem below presents the relationship between the RIP constants of $\mathcal{B}$ and $\mathcal{A}$. Set $\delta_1=\delta/(1-\delta)$.
\begin{theorem}\label{th.1}
Suppose that $\delta<1/2$ in (\ref{eq.9}) and that the linear measurement map $\mathcal{A}$ fulfills the RIP of order $r$ with constants $0<\mu_r\leq\nu_r$. It holds that the linear measurement map $\mathcal{B}$ obeys the RIP of order $r$ with constants $\mu_r(1-\delta_1)$ and $\nu_r(1-\delta_1)$.
\end{theorem}

\begin{remark}
The theorem shows that under the assumption of $\delta<1/2$ the RIP constants of $\mathcal{B}$ and $\mathcal{A}$ are equivalent.
\end{remark}

\begin{remark}
In the case of $m=n$ and the matrices $X=\mbox{diag}(x)~(x\in\mathbb{R}^m)$ is diagonal with $x$ being $r$-sparse (i.e., the number of non-zero elements in $x$ is $r$ at most), Theorem \ref{th.1} is the same as Proposition $1$ in \cite{Arias-Castro and Eldar}.
\end{remark}

\begin{proof}[Proof of the theorem \ref{th.1}: method 1] The idea is inspired by \cite{Arias-Castro and Eldar}. In order to bound $\|\Sigma_1-I_M\|$ we utilize the definition of $\delta$ in (\ref{eq.9}),
\begin{align}\label{eq.11}
\notag\|\Sigma_1-I_M\|&=\frac{\sigma^2_0mn}{\theta M}\|I_M-\frac{M}{mn}AA^{\top}\|\\
\notag&=\frac{\frac{\sigma^2_0mn}{M}}{\sigma^2+\frac{\sigma^2_0mn}{M}}\delta\\
&\leq \delta.
\end{align}
In the following, by applying the geometric series formula $1/(1-x)=\sum_{k=0}^{\infty}x^k$ for $|x|<1$, $\Sigma^{-1}_1-I$ is expressed as
\begin{align}\label{eq.12}
\Sigma^{-1}_1-I=[I-(I-\Sigma_1)]^{-1}-I=\sum_{k=1}^{\infty}(I-\Sigma_1)^k.
\end{align}
The above power series converges because $\|\Sigma_1-I\|\leq \delta<1$. In order to bound $\|\Sigma^{-1}_1-I\|$, we take operator norms on both sides of the equality (\ref{eq.12}), we get
\begin{align}\label{eq.13}
\notag\|\Sigma^{-1}_1-I\|&\overset{\text{(a)}}{\leq}\sum_{k=1}^{\infty}\|(I-\Sigma_1)^k\|\\
\notag&\overset{\text{(b)}}{\leq}\sum_{k=1}^{\infty}\|I-\Sigma_1\|^k\\
&\overset{\text{(c)}}{\leq}\sum_{k=1}^{\infty}\delta^k=\frac{\delta}{1-\delta}=:\delta_1,
\end{align}
where (a) follows from the triangle inequality, (b) uses the fact that $\|AB\|\leq\|A\|\|B\|$ for all matrices $A$ and $B$ in $\mathbb{R}^{M\times M}$, and (c) is due to (\ref{eq.11}).

Take any $X\in\mathbb{R}^{m\times n}$ satisfying rank at most $r$. Note that
\begin{align}\label{eq.14}
\notag&\|\mathcal{B}(X)\|^2_2-\|\mathcal{A}(X)\|^2_2\\
\notag&=\|B\mbox{vec}(X)\|^2_2-\|A\mbox{vec}(X)\|^2_2\\
&=\mbox{vec}^{\top}(X)A^{\top}(\Sigma^{-1}_1-I)A\mbox{vec}(X).
\end{align}
By employing H$\ddot{o}$lder's inequality, the definition of operator norm and (\ref{eq.13}), we get
\begin{align}\label{eq.15}
\notag&|\mbox{vec}^{\top}(X)A^{\top}(\Sigma^{-1}_1-I)A\mbox{vec}(X)|\\
\notag&\leq\|(A\mbox{vec}(X))^{\top}\|_2\|(\Sigma^{-1}_1-I)A\mbox{vec}(X)\|_2\\
\notag&\leq\|(\Sigma^{-1}_1-I)\|\|A\mbox{vec}(X)\|_2^2\\
&\leq\delta_1\|A\mbox{vec}(X)\|_2^2=\delta_1\|\mathcal{A}(X)\|_2^2.
\end{align}
A combination of (\ref{eq.14}) and (\ref{eq.15}), we get
\begin{align}\label{eq.16}
(1-\delta_1)\|\mathcal{A}(X)\|_2^2\leq\|\mathcal{B}(X)\|_2^2\leq(1+\delta_1)\|\mathcal{A}(X)\|_2^2.
\end{align}
Combining with (\ref{eq.16}) and (\ref{eq.10}), it implies
\begin{align}
\notag\mu_r(1-\delta_1)\|X\|^2_F\leq\|\mathcal{B}(X)\|_2^2\leq\nu_r(1+\delta_1)\|X\|^2_F.
\end{align}
This completes the proof.
\end{proof}

\begin{proof}[Proof of the theorem \ref{th.1}: method 2] We still use the preceding symbols unless specifically stated. Since vectorizing the matrix loses its structural information, we deal directly with it. Set $H=\mbox{Cov}(\mathcal{A}(Z))$. By some calculations,
\begin{align}
\notag H=\sigma^2_0\left[\begin{matrix}
\left<A^{(1)},A^{(1)}\right>~~\left<A^{(1)},A^{(2)}\right>~~\cdots~~\left<A^{(1)},A^{(M)}\right>\\
\left<A^{(2)},A^{(1)}\right>~~\left<A^{(2)},A^{(2)}\right>~~\cdots~~\left<A^{(2)},A^{(M)}\right>\\
\vdots\\
\left<A^{(M)},A^{(1)}\right>~~\left<A^{(M)},A^{(2)}\right>~~\cdots~~\left<A^{(M)},A^{(M)}\right>
\end{matrix}\right]:=\sigma^2_0G.
\end{align}
It follows that $\mbox{Cov}(v)=\sigma^2I_M+\sigma^2_0G:=\Sigma$. Set $\Sigma_1=\Sigma/\theta$ and
\begin{align}\label{eq.35}
\delta:=\left\|I_M-\frac{M}{mn}G\right\|.
\end{align}
It is easy to check that $\mbox{Cov}(u)=\mbox{Cov}(\Sigma^{-1/2}_1v)=\theta I_M$, i.e., $u$ is white noise.
Next we estimate the upper bound of $\|\Sigma_1-I_M\|$. By applying (\ref{eq.35}), we get
\begin{align}\label{eq.36}
\notag\|\Sigma_1-I_M\|&=\frac{\sigma^2_0mn}{\theta M}\|I_M-\frac{M}{mn}G\|\\
\notag&=\frac{\frac{\sigma^2_0mn}{M}}{\sigma^2+\frac{\sigma^2_0mn}{M}}\delta\\
&\leq \delta.
\end{align}
For any $r$-rank matrix $X$, due to $\mathcal{B}=\Sigma^{-1/2}_1\mathcal{A}$ and $\Sigma_1$ is symmetrical, we get
\begin{align}\label{eq.37}
\notag&\|\mathcal{B}(X)\|^2_2-\|\mathcal{A}(X)\|^2_2\\
\notag&=\|\Sigma^{-1/2}_1\mathcal{A}(X)\|^2_2-\|\mathcal{A}(X)\|^2_2\\
\notag&=\left<\Sigma^{-1/2}_1\mathcal{A}(X),\Sigma^{-1/2}_1\mathcal{A}(X)\right>-
\left<\mathcal{A}(X),\mathcal{A}(X)\right>\\
\notag&=\left<\mathcal{A}(X),(\Sigma^{-1/2}_1)^*\Sigma^{-1/2}_1\mathcal{A}(X)\right>-
\left<\mathcal{A}(X),\mathcal{A}(X)\right>\\
\notag&=\left<\mathcal{A}(X),\Sigma^{-1}_1\mathcal{A}(X)\right>-
\left<\mathcal{A}(X),\mathcal{A}(X)\right>\\
&=\left<\mathcal{A}(X),(\Sigma^{-1}_1-I)\mathcal{A}(X)\right>,
\end{align}
where $(\Sigma^{-1/2}_1)^*$ is the conjugate of $\Sigma^{-1/2}_1$. By using H$\ddot{o}$lder's inequality and (\ref{eq.13}), we get
\begin{align}\label{eq.38}
\notag&|\left<\mathcal{A}(X),(\Sigma^{-1}_1-I)\mathcal{A}(X)\right>|\\
\notag&\leq\|\mathcal{A}(X)\|_2\|(\Sigma^{-1}_1-I)\mathcal{A}(X)\|_2\\
\notag&\leq\|\Sigma^{-1}_1-I\|\|\mathcal{A}(X)\|_2^2\\
&\leq\delta_1\|\mathcal{A}(X)\|_2^2.
\end{align}
The remaining proof is the same as Method 1, and we omit it for brevity.
\end{proof}

Next, we present the concept of spherical section property of a linear measurement map.

The spherical section constant of a linear measurement map $\mathcal{A}$ is defined as
\begin{align}
\notag\Delta(\mathcal{A})=\min_{X\in\mathcal{N}(\mathcal{A})\backslash\{0\}}\frac{\|X\|^2_*}{\|X\|^2_F},
\end{align}
and we say $\mathcal{A}$ satisfies the $\Delta$-spherical section property (SSP) if $\Delta(\mathcal{A})\geq\Delta$, where $\|X\|_*$ is the nuclear norm of the matrix $X$, i.e., the sum of its singular values. In the following proposition, we will explore the connection between SSP constants of $\mathcal{A}$ and $\mathcal{B}$.
\begin{proposition}\label{pr.1}
Suppose that the linear measurement map $\mathcal{A}$ satisfies the $\Delta$-SSP with $\Delta>0$. Then the linear measurement map $\mathcal{B}$ obeys the $\Delta$-SSP with $\Delta$.
\end{proposition}

\begin{remark}
The proposition indicates that the SSP constants of $\mathcal{B}$ and $\mathcal{A}$ are identical.
\end{remark}

\begin{proof}[Proof of the lemma \ref{pr.1}]
Firstly, we show that $\mathcal{N}(\mathcal{A})=\mathcal{N}(\mathcal{B})$.

For any $X\in\mathcal{N}(\mathcal{A})\backslash\{0\}$, then $\mathcal{A}(X)=0$, i.e. $A\mbox{vec}(X)=0$. Note that $\mathcal{B}(X)=B\mbox{vec}(X)=\Sigma^{-1/2}A\mbox{vec}(X)$. Hence, $\mathcal{B}(X)=0$, namely, $X\in\mathcal{N}(\mathcal{B})\backslash\{0\}$. Therefore, $\mathcal{N}(\mathcal{A})\subseteq\mathcal{N}(\mathcal{B})$. Similarly, we could deduce that $\mathcal{N}(\mathcal{B})\subseteq\mathcal{N}(\mathcal{A})$. Combining with the above facts, $\mathcal{N}(\mathcal{A})=\mathcal{N}(\mathcal{B})$.

Now, we calculate the SSP constant of $\mathcal{B}$. By making use of the definition of SSP, we get
\begin{align}
\notag\Delta(\mathcal{B})&=\min_{X\in\mathcal{N}(\mathcal{B})\backslash\{0\}}\frac{\|X\|^2_*}{\|X\|^2_F}\\
\notag&=\min_{X\in\mathcal{N}(\mathcal{A})\backslash\{0\}}\frac{\|X\|^2_*}{\|X\|^2_F}\\
\notag&=\Delta(\mathcal{A})\geq\Delta.
\end{align}
The proof is complete.
\end{proof}

\section{The null space property for LRMR}
\label{sec3}

For recovering $X$, a prominent model is solving a constrained nuclear norm minimization problem
\begin{align}\label{eq.17}
\min_{\hat{X}\in\mathbb{R}^{m\times n}}\|\hat{X}\|_*~\mbox{subject~to}~\|\mathcal{B}(\hat{X})-\tilde{y}\|_2\leq \epsilon,
\end{align}
where $\epsilon=\sqrt{\theta(M+2\sqrt{M\log M})}$ stands for the noise level, and $\|\mathcal{B}(\hat{X})-\tilde{y}\|_2\leq \epsilon$ holds with high probability, for more details, see Lemma \ref{le.7}. As one of the crucial tool for the analysis of LRMR, the Frobenius-robust rank null space property (FRRNSP) of a linear measurement map $\mathcal{A}$ attracts specific interest.
\begin{definition} (FRRNSP \cite{Kabanava et al})\label{de.1}
The linear measurement map $\mathcal{A}:\mathbb{R}^{m\times n}\to\mathbb{R}^M$ is said to satisfy the Frobenius-robust rank null space property of order $r$ with constants $0<\rho<1$ and $\tau>0$ if for any $X\in\mathbb{R}^{m\times n}$, the singular values of $X$ fulfill
\begin{align}
\notag \|X_{[r]}\|_F\leq\frac{\rho}{\sqrt{r}}\|X_{[r]^c}\|_*+\tau\|\mathcal{A}(X)\|_2.
\end{align}
\end{definition}
Here, the singular value decomposition (SVD) of $X$ is $\sum_{i=1}^{n_0}\sigma_i(X)u_iv_i^{\top}$ with $n_0=\min\{m,n\}$, where $\sigma_i(X)$ is the $i$th largest singular value of $X$, and $u_i$ and $v_i$ are respectively the left and right singular value vectors of $X$. In this situation, write $X=X_{[r]}+X_{[r]^c}$, where $X_{[r]}$ is the best $r$-rank approximation of $X$, i.e., $X_{[r]}=\sum_{i=1}^r\sigma_i(X)u_iv_i^{\top}$. Combining with Definition \ref{de.1} and (\ref{eq.16}), we obtain the FRRNSP of the linear measurement map $\mathcal{B}$ given by the following lemma.
\begin{lemma}\label{le.1}
Set $\delta_1=\delta/(1-\delta)$. Under the assumptions of Definition \ref{de.1} and $\delta<1/2$, the linear measurement map $\mathcal{B}$ obeys the FRRNSP of order $r$, namely, for all $X\in\mathbb{R}^{m\times n}$,
\begin{align}
\notag \sqrt{1-\delta_1}\|X_{[r]}\|_F\leq\frac{\rho\sqrt{1-\delta_1}}{\sqrt{r}}\|X_{[r]^c}\|_*+\tau\|\mathcal{B}(X)\|_2,
\end{align}
holds for the singular values of $X$.
\end{lemma}
Based on the above notion and lemma, we will establish an FRRNSP condition for stable and robust recovery of low-rank matrix via the nuclear norm minimization and discuss the upper bound estimation of reconstruction error.

\begin{theorem}\label{th.2}
Suppose that a linear measurement map $\mathcal{A}:\mathbb{R}^{m\times n}\to\mathbb{R}^M$ satisfies the Frobenius-robust rank null space property of order $r$ with constants $0<\rho<1$ and $\tau>0$. Set $\delta_1=\delta/(1-\delta)$. Assume that $\delta<1/2$. Then, for any $X\in\mathbb{R}^{m\times n}$, a solution $X^*$ of (\ref{eq.17}) with $\tilde{y}=\mathcal{B}(X)+u$ and $\|u\|_2\leq\epsilon$ approximates the matrix $X$ with error
\begin{align}
\label{eq.18}\|X-X^*\|_F\leq C_1\frac{\|X_{[r]^c}\|_*}{\sqrt{r}}+C_2\epsilon,
\end{align}
where $$C_1=\frac{2(1+\rho)^2}{1-\rho}$$
and
$$C_2=\frac{2(3+\rho)\tau}{(1-\rho)\sqrt{1-\delta_1}}.$$
\end{theorem}

\begin{remark}
The theorem gives a sufficient condition to ensure the stable and robust reconstruction of the low-rank matrices.
\end{remark}

\begin{remark}
The inequality (\ref{eq.18}) in Theorem \ref{th.2} provides an
upper bound estimation on the reconstruction of the nuclear norm minimization. Especially, this estimation evidences that reconstruction precision of the nuclear norm minimization
can be controlled by the noise level and the best $r$-rank approximation error. Furthermore, the estimation (\ref{eq.18}) shows that the reconstruction accuracy of the method (\ref{eq.17}) can be bounded by the degree of rank of the matrix. In this sense, Theorem \ref{th.2} demonstrates that under certain conditions, an $r$-rank matrix can be robustly reconstructed by the method (\ref{eq.17}).
\end{remark}

\begin{remark}
When no noise is introduced, i.e., $u=0$ and $\epsilon=0$, it will result in the exact recovery when matrices are $r$-rank.
\end{remark}

\begin{remark}
By Lemma \ref{le.7}, we know that $u$ is Gaussian noise, so it is usually bounded by $l_2$-norm. However, when $u$ is non-Gaussian noise, for example, Gaussian mixture noise, it is more appropriate to exploit $l_p$-norm to bound that noise, see \cite{Wen et al 2017}. Then the real matrix could be robustly recovered by
\begin{align}
\label{eq.30}\min_{\hat{X}\in\mathbb{R}^{m\times n}}\|\hat{X}\|_*~\mbox{subject~to}~\|\mathcal{B}(\hat{X})-\tilde{y}\|_p\leq \epsilon,
\end{align}
where $p\geq1$, $\epsilon$ denotes the noise level which varies according to the range of $p$, and $\|\mathcal{B}(X)-\tilde{y}\|_p\leq \epsilon$ holds with high probability, for more details, see Lemma \ref{le.8}. In the following, we only consider the case of $1\leq p<2$ because of another (i.e. $p\geq2$) situation is similar. In this case, assuming the conditions of Lemma \ref{le.1} (just replace $\|\mathcal{B}(X)\|_2$ by $\|\mathcal{B}(X)\|_p$), the linear map $\mathcal{B}$ satisfies the FRRNSP of order $r$, viz,
\begin{align}
\notag \|X_{[r]}\|_F\leq\frac{\rho}{\sqrt{r}}\|X_{[r]^c}\|_*+\frac{\tau}{M^{1/2-1/p}\sqrt{1-\delta_1}}\|\mathcal{B}(X)\|_2.
\end{align}
Under the assumptions of Theorem \ref{th.2}, the solution $X^*$ of (\ref{eq.30}) satisfies
\begin{align}
\notag&\|X-X^*\|_F\leq \frac{2(1+\rho)^2}{1-\rho}\frac{\|X_{[r]^c}\|_*}{\sqrt{r}}
+\frac{2\tau(3+\rho)}{(1-\rho)M^{1/2-1/p}\sqrt{1-\delta_1}}\epsilon,\\
\notag&\|X-X^*\|_p\leq\frac{2(1+\rho)^2}{(1-\rho)r^{1-1/p}}\|X_{[r]^c}\|_*
+\frac{2\tau(3+\rho)r^{1/p-1/2}}{(1-\rho)M^{1/2-1/p}\sqrt{1-\delta_1}}\epsilon,
\end{align}
where $\epsilon=M^{1/p}\sqrt{\theta'(1+2\sqrt{M^{-1}\log M})}$ with $\theta'=[(1-\xi)+\kappa\xi]\sigma^2+mn[(1-\eta)+\gamma\eta]\sigma^2_0/M$.
\end{remark}

In the following, we present the stable rank null space property (SRNSP) of a linear measurement map weaker than the Frobenius-robust rank null space property, see Definition $4.17$ in \cite{Foucart and Rauhut} for the analogue in the sparse signal reconstruction situation.
\begin{definition} (SRNSP)\label{de.2}
We say that the linear measurement map $\mathcal{A}:\mathbb{R}^{m\times n}\to\mathbb{R}^M$ satisfies the stable rank null space property of order $r$ with constants $0<\rho<1$ and $\tau>0$ if for any $X\in\mathbb{R}^{m\times n}$, the singular values of $X$ fulfill
\begin{align}
\notag \|X_{[r]}\|_*\leq\rho\|X_{[r]^c}\|_*+\tau\|\mathcal{A}(X)\|_2.
\end{align}
\end{definition}
Similar to Lemma \ref{le.1}, we derive the following result on the SRNSP of the linear measurement map $\mathcal{B}$.
\begin{lemma}\label{le.2}
Set $\delta_1=\delta/(1-\delta)$. Assume that the conditions of Definition \ref{de.2} and $\delta<1/2$. Then, the linear measurement map $\mathcal{B}$ satisfies the SRNSP of order $r$, viz., for all $X\in\mathbb{R}^{m\times n}$,
\begin{align}
\notag \sqrt{1-\delta_1}\|X_{[r]}\|_*\leq\rho\sqrt{1-\delta_1}\|X_{[r]^c}\|_*+\tau\|\mathcal{B}(X)\|_2,
\end{align}
holds for the singular values of $X$.
\end{lemma}
With preparation above, we now state the stability and robustness of the method (\ref{eq.17}) under the definition scheme of SRNSP of a linear measurement map.
\begin{theorem}\label{th.3}
We assume that a linear measurement map $\mathcal{A}:\mathbb{R}^{m\times n}\to\mathbb{R}^M$ satisfies the sable rank null space property of order $r$ with constants $0<\rho<1$ and $\tau>0$. Set $\delta_1=\delta/(1-\delta)$ with $\delta<1/2$. Then, for any $X\in\mathbb{R}^{m\times n}$, a solution $X^*$ of (\ref{eq.17}) with $\tilde{y}=\mathcal{B}(X)+u$ and $\|u\|_2\leq\epsilon$ approximates the matrix $X$ with error
\begin{align}
\label{eq.19}\|X-X^*\|_F\leq D_1\frac{\|X_{[r]^c}\|_*}{\sqrt{r}}+D_2\epsilon,
\end{align}
where $$D_1=\frac{2(1+\rho)(\rho\sqrt{r}+1)}{1-\rho}$$
and
$$D_2=\frac{2[(1+\rho)\sqrt{r}+2]\tau}{(1-\rho)\sqrt{r(1-\delta_1)}}.$$
\end{theorem}

\begin{corollary}\label{co.1}
Under the same assumptions as in Theorem \ref{th.3}, suppose that $u=0$ and $X$ is $r$-rank. Then, $X$ can be exactly reconstructed via the method (\ref{eq.17}).
\end{corollary}

\section{Measurement map with independent entries and four finite moments}
\label{sec4}

In this section, we will determine how many measurement matrices with independent elements and four finite moments are needed for the FRRNSP condition to be fulfilled with high probability.
\begin{theorem}\label{th.4}
Set $n_0=\min\{m,n\}$. Let $\mathcal{A}:\mathbb{R}^{m\times n}\to\mathbb{R}^M$, and $\mathcal{A}(X)$ is defined by (\ref{eq.2}), where the $A^{(i)}$ are independent copies of a random matrix $\Phi=(X_{ij})_{i,j}$ with independent mean zero elements following
$$\mbox{Var}(X_{ij})=(\sigma^2+\sigma^2_0)/(\sigma^2+mn\sigma^2_0/M)$$
 and
 $$\mathbb{E}X^4_{ij}\leq(\sigma^2+\sigma^2_0)^2/(\sigma^2+mn\sigma^2_0/M)^2C_4$$
  for all $i,j$ and some positive constant $C_4$. In addition, assume that $A^{(1)},A^{(2)},\cdots,A^{(M)}$ are mutually orthogonal, and the columns of $A^{(i)}$ are mutually orthogonal and the lengths of its columns equal to $1/\sqrt{n}$.

\noindent  Then, for given $1\leq r\leq n_0$ and $\delta_1=\delta/(1-\delta)$ with $\delta<1/2$, there exists $c_1,c_2$ relying on $C_4$ that are positive constants, such that $\mathcal{A}$ satisfies the Frobenius-robust rank null space property with constants $0<\rho<1$ and $\tau>0$ with probability at least $1-e^{-c_2M}$ whenever $$M\geq c_1r(m+n).$$
\end{theorem}

\begin{proof}[Proof of the theorem \ref{th.4}] By the assumptions of Theorem \ref{th.4}, we get
\begin{align}
\notag &\mbox{vec}^{\top}(A^{(i)})\mbox{vec}(A^{(i)})=1,~i=1,2,\cdots,M,\\
\notag &\mbox{vec}^{\top}(A^{(i)})\mbox{vec}(A^{(j)})=0,~i\neq j.
\end{align}
Consequently,
\begin{align}
\notag AA^{\top}=\left(\begin{matrix}
\mbox{vec}^{\top}(A^{(1)})\\
\mbox{vec}^{\top}(A^{(2)})\\
\vdots\\
\mbox{vec}^{\top}(A^{(M)})
\end{matrix}\right)\left(\mbox{vec}(A^{(1)}),\mbox{vec}(A^{(2)}),\cdots,\mbox{vec}(A^{(M)})\right)=I_M.
\end{align}
By employing the identity above and the definition of $\Sigma_1$, we get
\begin{align}
\notag B_{i\cdot}&=(\Sigma^{-1/2}_1)_{i\cdot}A\\
\notag&=\left(\frac{\sigma^2+mn\sigma^2_0/M}{\sigma^2+\sigma^2_0}\right)^{1/2}A_{i\cdot},
\end{align}
i.e.,
$$\mbox{vec}^{\top}(B^{(i)})=\left(\frac{\sigma^2+mn\sigma^2_0/M}{\sigma^2+\sigma^2_0}\right)^{1/2}\mbox{vec}^{\top}(A^{(i)}).$$
By applying the conditions of Theorem \ref{th.4}, we get
$$\mathbb{E}(B_{ij})=\mathbb{E}\left[\left(\frac{\sigma^2+mn\sigma^2_0/M}{\sigma^2+\sigma^2_0}\right)^{1/2}A_{ij}\right]=0,$$
$$\mbox{Var}(B_{ij})=\mbox{Var}\left[\left(\frac{\sigma^2+mn\sigma^2_0/M}{\sigma^2+\sigma^2_0}\right)^{1/2}A_{ij}\right]=1,$$
$$\mathbb{E}(B^4_{ij})=\mathbb{E}\left[\left(\frac{\sigma^2+mn\sigma^2_0/M}{\sigma^2+\sigma^2_0}\right)^{1/2}A_{ij}\right]^4\leq C_4.$$
The remainder of the proof follows similarly the proof of Theorem $1.1$ in \cite{Kabanava et al}, which is omitted here for succinctness.
\end{proof}

\section{Numerical Simulations}
\label{sec5}

In this section, we present the optimization inside information of the constrained problem (\ref{eq.17}). The regularization form of the problem (\ref{eq.17}) is
\begin{align}\label{eq.32}
\min_{\hat{X}}\|\hat{X}\|_*+\frac{\lambda}{2}\|B\mbox{vec}(\hat{X})-\tilde{y}\|^2_2,
\end{align}
where $\lambda$ is a regularization parameter, $\hat{X}\in\mathbb{R}^{m\times n}$, $B\in\mathbb{R}^{M\times mn}$, $\tilde{y}\in\mathbb{R}^M$ and $\mbox{vec}(\hat{X})$ stands for the vectorization of $\hat{X}$. Then, we solve the unconstrained problem (\ref{eq.32}) by using the alternating direction method of multipliers (ADMM) \cite{Lu et al arXiv 2018} \cite{Lu et al TPAMI 2018} \cite{Feng et al 2019}. The problem (\ref{eq.32}) can be equivalently rewritten as
\begin{align}\label{eq.33}
\min_{\hat{X}}\|\hat{X}\|_*+\frac{\lambda}{2}\|B\mbox{vec}(U)-\tilde{y}\|^2_2~\mbox{s.t.}~\hat{X}=U.
\end{align}
The associating augmented Lagrangian function is
\begin{align}\label{eq.34}
L(\hat{X},U,W)=\|\hat{X}\|_*+\frac{\lambda}{2}\|B\mbox{vec}(\hat{X})-\tilde{y}\|^2_2
+\left<W,\hat{X}-U\right>+\frac{\rho_1}{2}\|\hat{X}-U\|^2_F.
\end{align}
where $W\in\mathbb{R}^{m\times n}$ indicates the Lagrangian multiplier, and $\rho_1$ is a positive scalar. Then $\hat{X}$ and $W$ can be obtained by minimizing each variable alternately while fixing the other variables. The updated details are summarized in Algorithm 5.1.

\begin{algorithm}[H]
\label{alg.1}\caption{: Solve problem (\ref{eq.17}) by ADMM } 
\begin{algorithmic}[1]
\State Input $A\in\mathbb{R}^{M\times mn}$, $y\in\mathbb{R}^M$, $\sigma,~\sigma_0$.
\State Whitening $B=\Sigma^{-1/2}_1A$, $\tilde{y}=\Sigma^{-1/2}_1y$.
\State Initialize $\hat{X}^0=U^0=W^0$, $\gamma=1.1$, $\lambda_0 = 10^{-6}$, $\lambda_{\max} = 10^{10}$, $\rho_1 = 10^{-6}$, $\varepsilon=10^{-8}$, $j=0$.
\While{not converged} 
¡¡¡¡\State Updated $X^{j+1}$ by $$\hat{x}=(B^{\top}B+\rho_1I)^{-1}\left(B^{\top}\tilde{y}-\rho_1\mbox{vec}(U^j)+\mbox{vec}(W^j)\right);$$
     $\hat{X}^{j+1}\leftarrow\hat{x}$: reshape $\hat{x}$ to the matrix $\hat{X}^{j+1}$ of size $m\times n$.
    \State Update $U^{j+1}$ by
    $$\arg\min_U\rho_1\left(\frac{\lambda}{\rho_1}\|U\|_*+\frac{1}{2}\left\|U-\left(X^{j+1}+\frac{W^j}{\rho_1}\right)\right\|^2_F\right);$$
     \State Update $W^{j+1}$ by
     $$W^{j+1}=W^{j}+\hat{X}^{j+1}-U^{j+1};$$
     \State Update $\lambda_{j+1}$ by $\lambda_{j+1}=\min(\gamma\lambda_{j},\lambda_{\max});$
     \State Check the convergence conditions
      $$\|\hat{X}^{j+1}-X^j\|_{\infty}\leq\varepsilon,\|U^{j+1}-U^j\|_{\infty}\leq\varepsilon,$$
      $$\|B\mbox{vec}(\hat{X}^{j+1})-\tilde{y}\|_{\infty}\leq\varepsilon,\|\hat{X}^{j+1}-U^{j+1}\|_{\infty}\leq\varepsilon.$$

\EndWhile
\end{algorithmic}
\end{algorithm}

In our experiments, the measurement matrix $A\in\mathbb{R}^{M\times mn}$ is generated with its elements being i.i.d., zero-mean, $1/M$-variance Gaussian distribution. Next, the matrix $X\in\mathbb{R}^{m\times n}$ of rank $r$ is generated by $X=XL*XR$, where $XL\in\mathbb{R}^{m\times r}$ and $XL\in\mathbb{R}^{r\times n}$ are with i.i.d. draw from a standard Gaussian distribution. The noise matrix $Z$ and the measurement noise vector $w$ are then respectively generated with its entries being i.i.d., zero-mean, $\sigma^2_0$-variance Gaussian distribution ($\sigma_0=0.05,~0.10,~0.15$) and $\sigma^2$-variance Gaussian distribution ($\sigma=0.01$). We choose $m=n=30$ and $r=0.2m$. With $A$, $X$, $Z$ and $w$, the measurement $y$ is produced by $y=A(\mbox{vec}(X)+\mbox{vec}(Z))+w$. Due to $\theta=\sigma^2+mn\sigma^2_0/M$, $\Sigma_1=\Sigma/\theta$, accordingly we derive $\tilde{y}=\Sigma^{-1/2}_1y$, $B=\Sigma^{-1/2}_1A$ after whitening noise. In order to prevent the occurrence of randomness, we reveal the average results over independent 100 trails in all experiments.

To find the better $\lambda$ which derives the maximal Signal-to-Noise Ratio (SNR, $20\log(\|X\|_F/\|X-\hat{X}\|_F)$), a set of trails have been carried out. In Fig. \ref{fig.4} with $M=750$, the SNR is plotted versus the regularization parameter $\lambda$ for different $\sigma_0$ values, $\sigma_0=0.05,~0.10,~0.15$, and $\lambda$ is varied between $10^{-9}$ and $1$, and the image evidences that the parameter $\lambda\in[10^{-9},10^{-1}]$ is a well selection.

\begin{figure}[h]
\begin{center}
\includegraphics[width=0.45\textwidth]{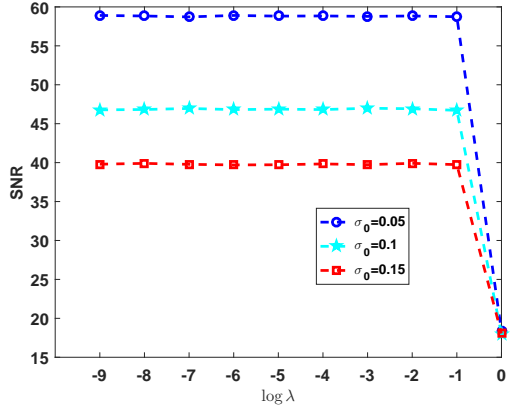}
\caption{Recovery performance of the constrained nuclear norm minimization (\ref{eq.17}) versus $\lambda$.}\label{fig.4}
\end{center}
\vspace*{-14pt}
\end{figure}

In order to verify the justifiability of the model (\ref{eq.17}), two sets of experiments have been conducted. In Fig. \ref{fig.1}(a) with $M=700$, the average relative error ($\|X-\hat{X}\|_F/\|X\|_F$) is plotted versus the rank $r$ for different standard deviation values (i.e., $\sigma_0$), $\sigma_0=0.05,~0.10,~0.15$, and the rank $r$ ranges from $4$ and $8$ and Fig. \ref{fig.1}(b) depicts the SNR versus the number of measurements $M$ for different $\sigma_0$ values, $\sigma_0=0.05,~0.10,~0.15$, and the number of measurements $M$ varies from $720$ to $800$ with $r=6$. It is easy to see that as the rank of the original matrix $X$decreases and the number of measurements increases, the recovery error decreases gradually, and a decreasing standard deviation $\sigma_0$ leads to a better performance.

\begin{figure}[h]
\begin{center}
\subfigure[]{\includegraphics[width=0.40\textwidth]{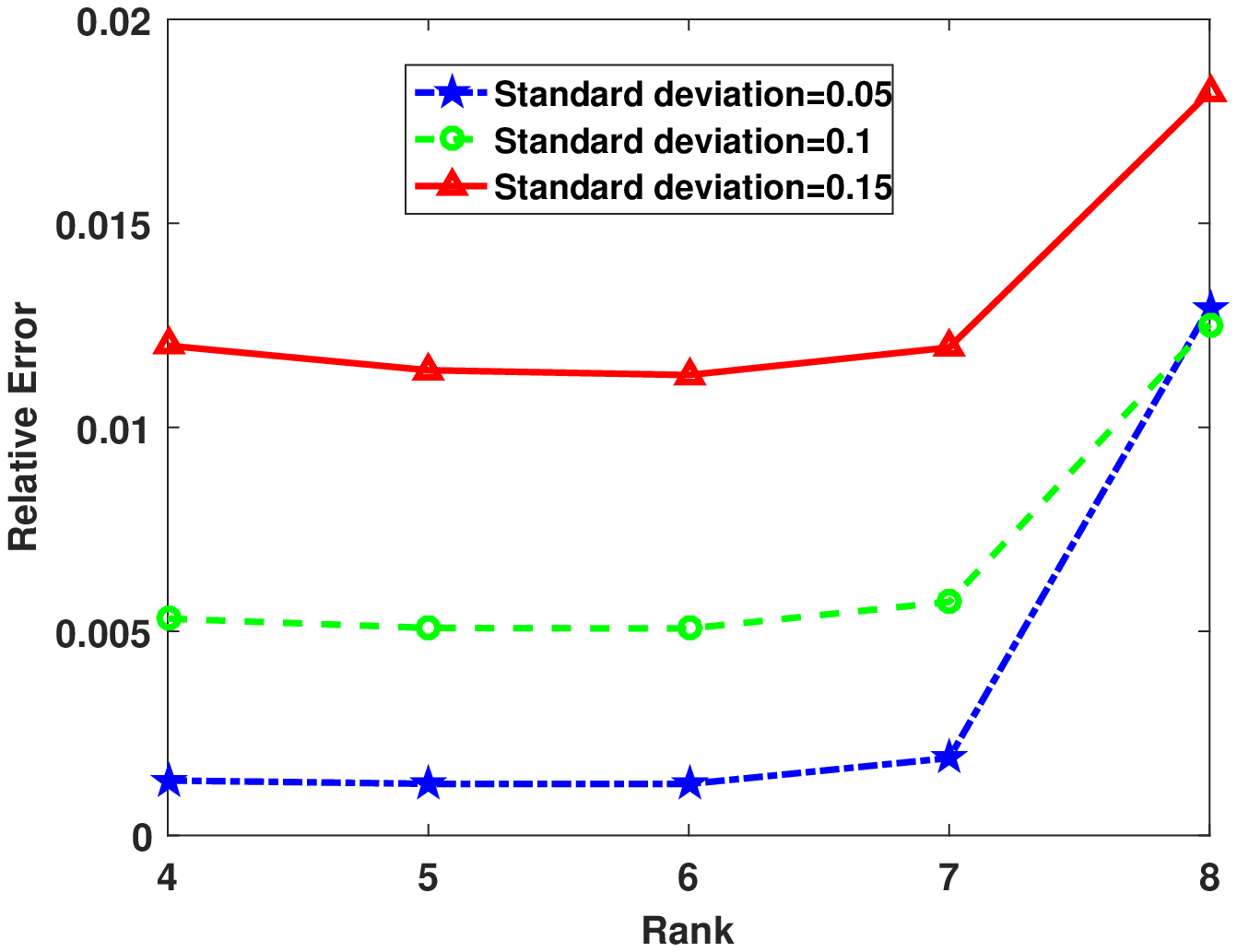}}
\subfigure[]{\includegraphics[width=0.40\textwidth]{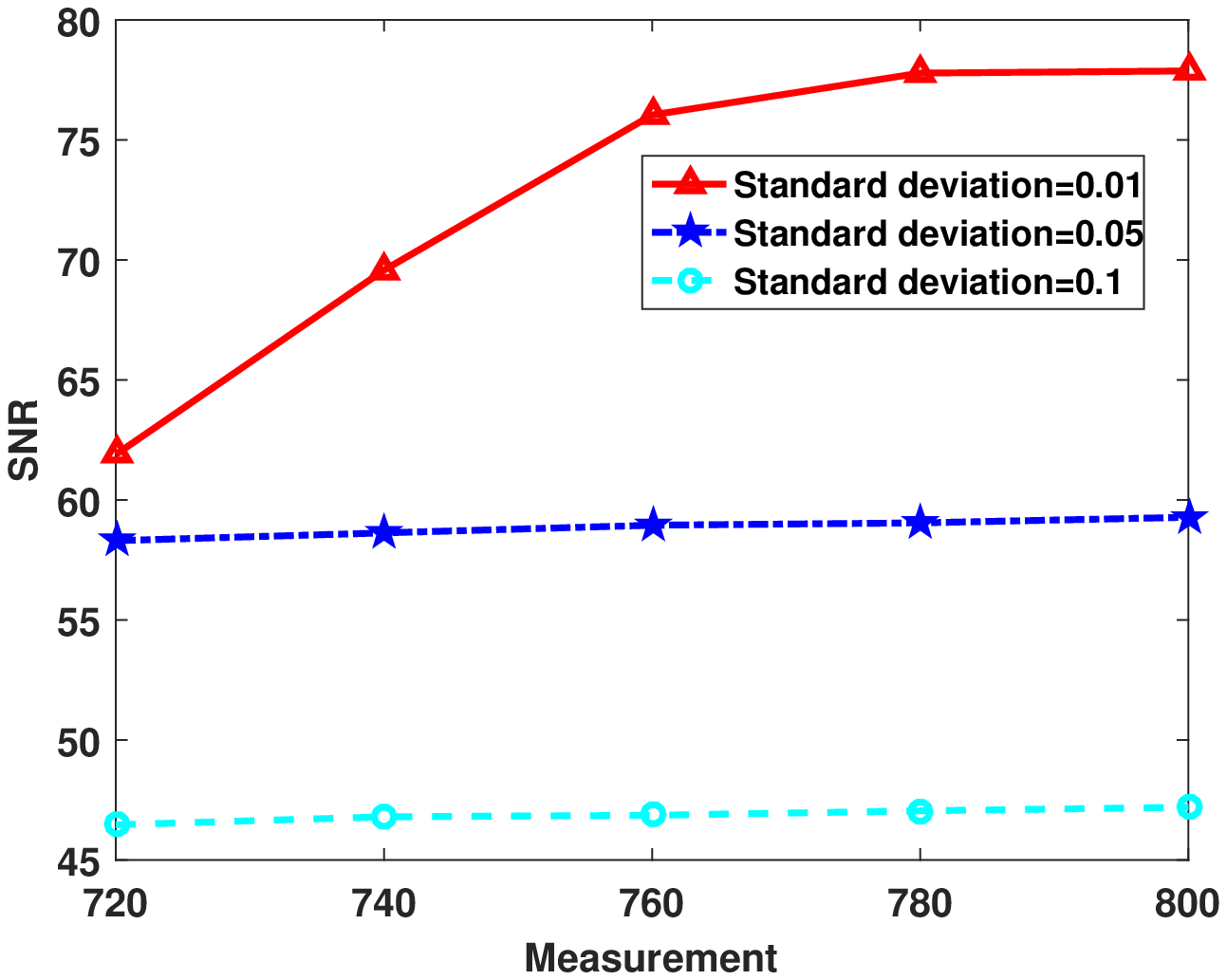}}
\caption{(a) Relative error varying rank with $M=700$, (b) SNR varying number of measurements with $r=6$.}\label{fig.1}
\end{center}
\vspace*{-14pt}
\end{figure}

To further verify the validity of the model (\ref{eq.17}) recovery, we choose random Bernoulli matrix as the measurement matrix, whose entries follows Bernoulli distribution, i.e.,
\[  A_{ij}=  \frac{1}{\sqrt{M}}\begin{cases}1,\quad \ \ &    p=\frac{1}{2},\\    -1, \quad \ \ & p=\frac{1}{2}.  \end{cases}  \]

SNR versus the rank and the relative error versus the number of samples, the results are shown in Fig. \ref{fig.2}(a) and (b) in different $\sigma_0=0.05,~0.10,~0.15$. In Fig. \ref{fig.2}(a), the values of the rank $r$ of the original matrix vary from $4$ to $8$ with $M=700$ and in Fig. \ref{fig.2}(b), the number of samples $M$ ranges from $730$ to $810$ with $r=6$. Fig. \ref{fig.2}(a) and (b) demonstrate that as the variance of noise matrix $Z$ decreases, the recovery effect becomes better, and a smaller rank of the original matrix and a larger number of samples make the reconstruction error smaller (SNR larger).

\begin{figure}[h]
\begin{center}
\subfigure[]{\includegraphics[width=0.40\textwidth]{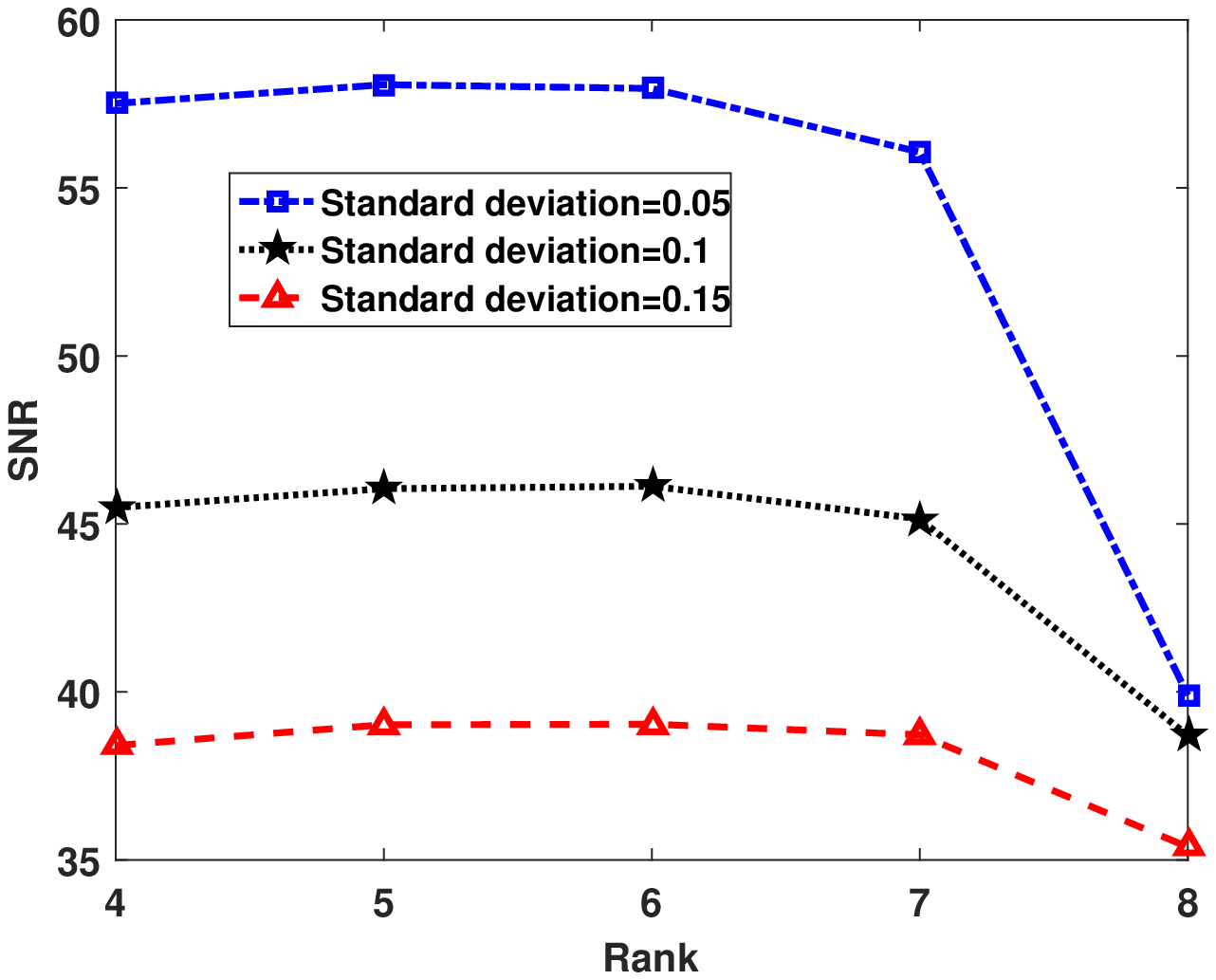}}
\subfigure[]{\includegraphics[width=0.40\textwidth]{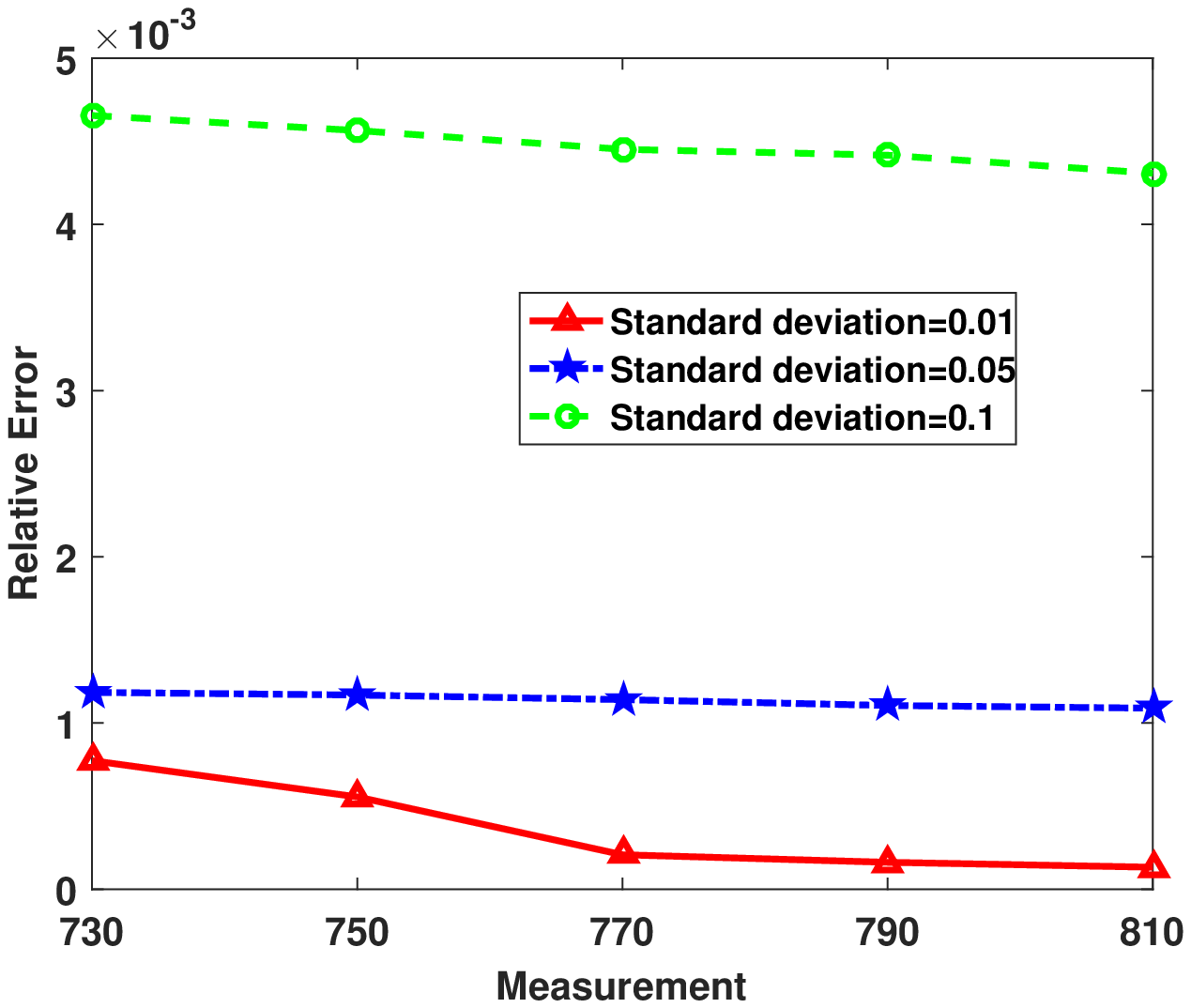}}
\caption{(a) SNR varying rank with $M=700$, (b) Relative error varying number of samples with $r=6$.}\label{fig.2}
\end{center}
\vspace*{-14pt}
\end{figure}

Finally, the effect of noise variance on the performance of model (\ref{eq.17}) reconstruction is illustrated by grayscale image recovery. The original image (in Fig. \ref{fig.3}) has a resolution of $256\times256$. The selection of measurement matrix is the same as that in Fig. \ref{fig.1}. We fix the standard deviation $\sigma=0.01$ of the measurement noise $w$. Due to the limitation of experimental conditions, we scale the original image to a resolution of $30\times30$. The number of measurements is equal to $M=2.5r(m+n-r)+1$ and the rank $r$ equals to $0.2m$. To access the quality of the recovered image, we adopt Structural SIMilarity (SSIM) and Peak Signal-to-Noise Ratio (PSNR). The gained results are reported in Table \ref{tab.1}. The results show again the smaller the variance of noise, the better the recovery effect.

\begin{figure}[h]
\begin{center}
\includegraphics[width=0.40\textwidth]{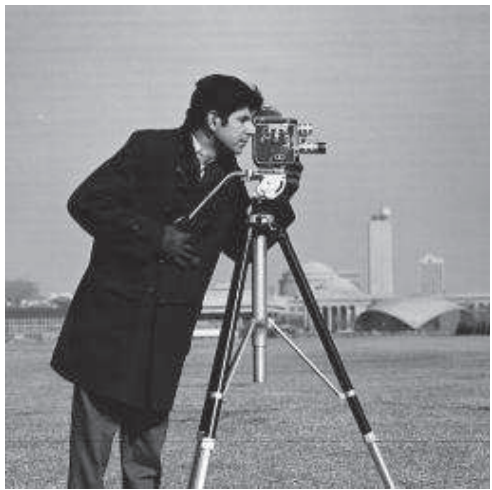}
\caption{Grayscale image of $256\times256$ pixels.}\label{fig.3}
\end{center}
\vspace*{-14pt}
\end{figure}

\makeatletter\def\@captype{table}\makeatother
\caption{PSNR$|$SSIM results on recovery of test image}\label{tab.1}
  \begin{center}
  \begin{tabular}{|l|l|l|l|}
\hline
  & $\sigma_0$ & PSNR & SSIM \\
\hline
 1 & 0.05 & 35.588 & 0.95328 \\
\hline
 2 & 0.10 & 33.248 & 0.92672 \\
\hline
3 & 0.15 & 30.142 & 0.86879 \\
\hline
4 & 0.20 & 27.033 & 0.77871 \\
\hline
\end{tabular}
\end{center}

Some experts may ask, what is the effect of the model and its associating algorithm on image denoising? We have made some attempts in this respect, but we have not yet achieved good experimental results. We are still trying to explore this issue and regard it as an important research direction in the future.

\section{The proofs of theorems}
\label{sec6}

Before proving our main results, we need some auxiliary lemmas.
\begin{lemma} (Gaussian white noise) \label{le.7}
Recall that $w$ is a white noise vector with $\mathbb{E}(w)=0_M$ and $\mbox{Var}(w)=\sigma^2I_M$, and that similarly $Z$ is a white noise matrix satisfying $\mathbb{E}(Z)=0_{m\times n}$ and $\mbox{Var}(Z)=\sigma^2_0I_{mn}$, independent of $w$. In addition, we assume that both $w$ and $Z$ follow Gaussian distribution. Then
\begin{align}
\notag u=\Sigma^{-1/2}_1v\sim\mathbb{N}(0,\theta I_M),
\end{align}
where $\Sigma_1=\Sigma/\theta$, $\Sigma=\sigma^2I_M+\sigma^2_0AA^{\top}$, $\theta=\sigma^2+mn\sigma^2_0/M$, and $v=A\mbox{vec}(Z)+w$.

Furthermore, the noise $u\sim\mathbb{N}(0,\theta I_M)$ satisfies
\begin{align}\label{eq.26}
\mathbb{P}\left(\|u\|_2\leq\sqrt{\theta(M+2\sqrt{M\log M})}\right)\geq 1-\frac{1}{M},
\end{align}
\begin{align}\label{eq.27}
\begin{cases}
\mathbb{P}\left(\|u\|_p\leq M^{1/p}\sqrt{\theta(1+2\sqrt{M^{-1}\log M})}\right)> 1-\frac{1}{M},~0\leq p<2,\\
\mathbb{P}\left(\|u\|_p\leq\sqrt{\theta(M+2\sqrt{M\log M})}\right)\geq 1-\frac{1}{M},~p\geq2.
\end{cases}
\end{align}
\end{lemma}

\begin{proof}[Proof of the lemma \ref{le.7}]
Note that $A$ is a linear transformation. By applying the property of Gaussian random vector, the definition of covariance matrix and some elementary calculations, we get
\begin{align}
\notag v\sim \mathbb{N}(0,\sigma^2I_M+\sigma^2_0AA^{\top}).
\end{align}
Since $\Sigma^{-1/2}_1$ is a linear map, we get
\begin{align}
\notag u=\Sigma^{-1/2}_1v\sim\mathbb{N}(0,\theta I_M).
\end{align}

The proof of the inequality (\ref{eq.26}) is similar to Lemma $III.3$ in \cite{Lin et al}.

In the following, we prove the equation (\ref{eq.27}). Since the proofs of two cases are similar, we only provide the proof of $0\leq p<2$. By employing the inequality that $\|x\|_2\leq\|x\|_p\leq M^{1/p-1/2}\|x\|_2$ for all $x\in\mathbb{R}^M$ and fixed $1\leq p<2$ and $\|u\|_2\leq\sqrt{\theta(M+2\sqrt{M\log M})}$, we get
$$\|u\|_p\leq M^{1/p-1/2}\sqrt{\theta(M+2\sqrt{M\log M})},$$
which implies
\begin{align}
\notag\mathbb{P}\left(\|u\|_p\leq M^{1/p}\sqrt{\theta(1+2\sqrt{M^{-1}\log M})}\right)>\mathbb{P}\left(\|u\|_2\leq\sqrt{\theta(M+2\sqrt{M\log M})}\right)\geq 1-\frac{1}{M}.
\end{align}
\end{proof}

\begin{lemma} (Gaussian mixture noise) \label{le.8}
Assume that i.i.d. $w_i$ and $Z_{ij}$ follow respectively two-term Gaussian mixture models, i.e., $w_i\sim(1-\xi)\mathbb{N}(0,\sigma^2)+\xi\mathbb{N}(0,\kappa\sigma^2),~i=1,\cdots,M$, and $Z_{ij}\sim(1-\eta)\mathbb{N}(0,\sigma^2_0)+\eta\mathbb{N}(0,\gamma\sigma^2_0),~i=1,\cdots,m,~j=1,\cdots,n$, where $0\leq\xi<1~(0\leq\eta<1)$ represents the portion of outliers in the noise and $\kappa>1~(\gamma>1)$ stands for
the strength of outliers. Then
\begin{align}
\notag u=\Sigma'^{-1/2}_1v\sim\mathbb{N}(0,\theta' I_M),
\end{align}
namely, $u_i$ obeys the Gaussian mixture noise, i.e., $u_i\sim(1-\xi)\mathbb{N}(0,\sigma^2)+\xi\mathbb{N}(0,\kappa\sigma^2)+(1-\eta)\mathbb{N}(0,mn\sigma^2_0/M)
+\eta\mathbb{N}(0,mn\gamma\sigma^2_0/M)$,
where $\Sigma'_1=\Sigma'/\theta'$, $\Sigma'=[(1-\xi)+\kappa\xi]\sigma^2I_M+[(1-\eta)+\gamma\eta]\sigma^2_0AA^{\top}$, $\theta'=[(1-\xi)+\kappa\xi]\sigma^2+mn[(1-\eta)+\gamma\eta]\sigma^2_0/M$, and $v=A\mbox{vec}(Z)+w$.

Besides, the noise $u\sim\mathbb{N}(0,\theta' I_M)$ fulfills
\begin{align}\label{eq.28}
\begin{cases}
\mathbb{P}\left(\|u\|_p\leq M^{1/p}\sqrt{\theta'(1+2\sqrt{M^{-1}\log M})}\right)> 1-\frac{1}{M},~0\leq p<2,\\
\mathbb{P}\left(\|u\|_p\leq\sqrt{\theta'(M+2\sqrt{M\log M})}\right)\geq 1-\frac{1}{M},~p\geq2,
\end{cases}
\end{align}
where $\theta'=[(1-\xi)+\kappa\xi]\sigma^2+mn[(1-\eta)+\gamma\eta]\sigma^2_0/M$. And similar to \cite{Wen et al 2017}, the $p$-th moment of such a noise process is given by
\begin{align}\label{eq.29}
\mathbb{E}\{\|u\|_p^p\}=\frac{M2^{p/2}\theta'^{p/2}\Gamma(\frac{p+1}{2})}{\sqrt{\pi}}
\end{align}
where $\theta'=[(1-\xi)+\kappa\xi]\sigma^2+mn[(1-\eta)+\gamma\eta]\sigma^2_0/M$.
\end{lemma}

The following lemma presents a matrix of Stechkin's bound generalizing the result of sparse vectors \cite{Foucart and Rauhut} to the case of low-rank matrices.

\begin{lemma} (\cite{Kabanava et al})\label{le.3}
Let $X\in\mathbb{R}^{m\times n}$ and $r\leq\min\{m,n\}$. Then, for $p>0$,
\begin{align}
\notag\|X_{[r]^c}\|_p\leq\frac{1}{r^{1-1/p}}\|X\|_*.
\end{align}
\end{lemma}
The following lemma is a useful inequality on matrix norm.
\begin{lemma} (\cite{Horn and Johnson})\label{le.4}
For any $X, Y\in\mathbb{R}^{m\times n}$, we have
\begin{align}
\notag\|X-Y\|_*\geq\sum^{n_0}_{i=1}|\sigma_i(X)-\sigma_i(Y)|,
\end{align}
where $n_0=\min\{m,n\}$.
\end{lemma}

The result below reveals that the distance between the original matrix and its corresponding solution is bounded by the best $r$-rank approximation error and the Euclidean norm of the difference between their measurements provided that the Frobenius-robust rank null space property.

\begin{lemma}\label{le.5}
Set $\delta_1=\delta/(1-\delta)$ with $\delta<1/2$. Assume that $\mathcal{A}:\mathbb{R}^{m\times n}\to\mathbb{R}^M$ fulfills the Frobenius-robust rank null space property with constants $0<\rho<1$ and $\tau>0$. Then, a solution $X^*$ of problem (\ref{eq.17}) approximates the matrix $X$ with errors
\begin{align}
\label{eq.20}\|X-X^*\|_*\leq \frac{2(1+\rho)}{1-\rho}\|X_{[r]^c}\|_*+\frac{2\tau\sqrt{r}}{(1-\rho)\sqrt{1-\delta_1}}\|\mathcal{B}(X-X^*)\|_2.
\end{align}
\end{lemma}

\begin{proof}[Proof of the lemma \ref{le.5}] By exploiting Lemma \ref{le.4}, we get
\begin{align}
\notag\|X^*\|_*&=\|X-(X-X^*)\|_*\geq\sum^{n_0}_{i=1}|\sigma_i(X)-\sigma_i(X-X^*)|\\
\notag&=\sum^{r}_{i=1}|\sigma_i(X)-\sigma_i(X-X^*)|+\sum^{n_0}_{i=r+1}|\sigma_i(X)-\sigma_i(X-X^*)|\\
\notag&\geq\sum^{r}_{i=1}(\sigma_i(X)-\sigma_i(X-X^*))+\sum^{n_0}_{i=r+1}(\sigma_i(X-X^*)-\sigma_i(X)).
\end{align}
Therefore,
\begin{align}
\notag\|(X-X^*)_{[r]^c}\|_*&\leq\|X^*\|_*-\sum^{r}_{i=1}\sigma_i(X)+\sum^{r}_{i=1}\sigma_i(X-X^*)
+\sum^{n_0}_{i=r+1}\sigma_i(X)\\
\notag&\overset{\text{(a)}}{\leq}\|X^*\|_*-\|X\|_*+\sqrt{r}\|(X-X^*)_{[r]}\|_F+2\|X_{[r]^c}\|_*\\
\notag&\overset{\text{(b)}}{\leq}\sqrt{r}\|(X-X^*)_{[r]}\|_F+2\|X_{[r]^c}\|_*,
\end{align}
where (a) follows from H$\ddot{o}$lder's inequality, and (b) is due to the minimality of $X^*$. Employing the Frobenius-robust rank null space property of the linear measurement map $\mathcal{B}$ to the inequality above, rearranging the terms and observing $0<\rho<1$, we get
\begin{align}
\label{eq.21}\|(X-X^*)_{[r]^c}\|_*\leq\frac{1}{1-\rho}\left(\frac{\tau\sqrt{r}}{\sqrt{1-\delta_1}}\|\mathcal{B}(X-X^*)\|_2
+2\|X_{[r]^c}\|_*\right).
\end{align}
Then,
\begin{align}
\notag\|X-X^*\|_*&=\|(X-X^*)_{[r]}\|_*+\|(X-X^*)_{[r]^c}\|_*\\
\notag&\overset{\text{(a)}}{\leq}\sqrt{r}\|(X-X^*)_{[r]}\|_F
+\|(X-X^*)_{[r]^c}\|_*\\
\notag&\overset{\text{(b)}}{\leq}(1+\rho)\|(X-X^*)_{[r]^c}\|_*
+\frac{\tau\sqrt{r}}{\sqrt{1-\delta_1}}\|\mathcal{B}(X-X^*)\|_2\\
\notag&\overset{\text{(c)}}{\leq}\frac{2(1+\rho)}{1-\rho}\|X_{[r]^c}\|_*
+\frac{2\tau\sqrt{r}}{(1-\rho)\sqrt{1-\delta_1}}\|\mathcal{B}(X-X^*)\|_2,
\end{align}
where (a) is from H$\ddot{o}$lder's inequality, (b) is due to the Frobenius-robust rank null space property of $\mathcal{B}$, and (c) follows from the inequality (\ref{eq.21}). The proof is complete.
\end{proof}
\begin{proof}[Proof of the theorem \ref{th.2}] By utilizing Lemma \ref{le.3}, we get
\begin{align}
\notag \|X-X^*\|_F&\overset{\text{(a)}}{\leq}\|(X-X^*)_{[r]}\|_F+\|(X-X^*)_{[r]^c}\|_F\\
\notag&\leq\|(X-X^*)_{[r]}\|_F+\frac{1}{\sqrt{r}}\|X-X^*\|_*\\
\notag&\overset{\text{(b)}}{\leq}\frac{1+\rho}{\sqrt{r}}\|X-X^*\|_*+\frac{\tau}{\sqrt{1-\delta_1}}\|\mathcal{B}(X-X^*)\|_2,
\end{align}
where for (a) we make use of the triangular inequality, and (b) follows from the Frobenius-robust rank null space property of $\mathcal{B}$. Plugging (\ref{eq.20}) into the above inequality, observing the fact that $\|u\|_2\leq\epsilon$, and due to the inequality $\|\mathcal{B}(X-X^*)\|_2\leq\|\mathcal{B}(X)-\tilde{y}\|_2+\|\mathcal{B}(X^*)-\tilde{y}\|_2$, the result is deduced.
\end{proof}

The following outcome clears that under the stable rank null space property, the distance between the matrix to be recovered and its associating solution is controlled by the best $r$-rank approximation error and the Euclidean distance between their measurements.

\begin{lemma}\label{le.6}
Set $\delta_1=\delta/(1-\delta)$ with $\delta<1/2$. Let $X^*$ be the optimal solution of problem (\ref{eq.17}) with $\|\mathcal{B}(X^*)-\tilde{y}\|_2\leq \epsilon$. Suppose that we observe $\tilde{y}=\mathcal{B}(X)+u$ with $\|\mathcal{B}(X)-\tilde{y}\|_2\leq \epsilon$ and $\mathcal{A}:\mathbb{R}^{m\times n}\to\mathbb{R}^M$ meets the stable rank null space property with constants $0<\rho<1$ and $\tau>0$. Then,
\begin{align}
\label{eq.22}\|X-X^*\|_*\leq \frac{2(1+\rho)}{1-\rho}\|X_{[r]^c}\|_*+\frac{2\tau}{(1-\rho)\sqrt{1-\delta_1}}\|\mathcal{B}(X-X^*)\|_2.
\end{align}
\end{lemma}

\begin{proof}[Proof of the lemma \ref{le.6}]
The proof is similar to that of Lemma \ref{le.5}.
\end{proof}

\begin{proof}[Proof of the theorem \ref{th.3}] Combining with the stable rank null space property, Cauchy-Schwarz inequality, Lemmas \ref{le.3} and \ref{le.6}, the desired result is established. This completes the proof.
\end{proof}

\section{Conclusion}
\label{sec7}

\noindent Although the literature on low-rank matrix recovery is almost silent on the impact of pre-measurement noise on recovery performance, this paper certificated that it maybe have an important effect on signal-noise-ratio. Certainly, we indicated that for the widespread measuring formula employed in low-rank matrix recovery, the model with pre-measurement noise is, after whitening, equivalent to a standard model with merely additional noise and a raise in the noise variance by a factor of $mn/M$. We presented bounds on the RIP constants and the SSR constants of new linear measurement map which cleared that as $m,~n,~M\to\infty$ with $mn/M\to 0$, the RIP constants are fundamentally unaltered. As the performance of standard reconstruction approaches is regularly expressed with respect to the RIP constants, this demonstrates that, these approaches manipulate like the standard, as well as noise folding causes a large noise increase. Besides, based on the two kinds of null space properties, we extended the study to the noise folding scenario, established sufficient conditions for robustly reconstructing low-rank matrix itself subject to noise, and provided upper bound estimations of recovery error. Furthermore, the minimal number of measurement such that sufficient condition based on FRRNSP obeys was gained. Numerical simulations are presented to verify the theoretical results.

\section{Acknowledgments}

\noindent This work was supported by Natural Science Foundation of China (Nos. 61673015, 61273020), Fundamental Research Funds for the Central
Universities (Nos. XDJK2015A007, XDJK2018C076, SWU1809002), Youth Science and technology talent development project (No. Qian jiao he KY zi [2018]313), Science and technology Foundation of Guizhou province (No. Qian ke he Ji Chu [2016]1161), Guizhou province natural science foundation in China (No. Qian Jiao He KY [2016]255).

\end{document}